\documentclass[preprint,11pt]{elsarticle}
\usepackage[margin=1in]{geometry}

\usepackage{amssymb}
\usepackage{amsthm}
\usepackage{amsmath}
\usepackage{mathrsfs}
\usepackage[dvipsnames]{xcolor}
\usepackage{hyperref}
  \hypersetup{
    colorlinks=true,
    linkcolor=PineGreen,
    filecolor=magenta,      
    urlcolor=cyan}

\journal{}

\theoremstyle{definition}
\newtheorem{definition}{Definition}

\theoremstyle{plain}
\newtheorem{lemma}{Lemma}
\newtheorem{proposition}{Proposition}
\newtheorem{theorem}{Theorem}

\theoremstyle{remark}

\begin{document}

\begin{frontmatter}



\title{The Nature of Organization in Living Systems}


\author[lab1]{Pedro Márquez-Zacarías\fnref{label1}}
\author[lab1]{Andrés Ortiz-Muñoz\fnref{label11}}
\author[lab2,lab3]{and Emma P. Bingham}

\affiliation[lab1]{organization={Santa Fe Institute},
            addressline={1399 Hyde Park Rd}, 
            city={Santa Fe},
            postcode={87501}, 
            state={NM},
            country={USA}}
\affiliation[lab2]{organization={School of Physics, Georgia Institute of Technology}, 
            city={Atlanta},
            postcode={30332}, 
            state={GA},
            country={USA}}
\affiliation[lab3]{organization={Interdisciplinary Program in Quantitative Biosciences, Georgia Institute of Technology}, 
            city={Atlanta},
            postcode={30332}, 
            state={GA},
            country={USA}}

\fntext[label1]{These authors contributed equally. Correspondence: pedromaz@santafe.edu}

\begin{abstract}
Living systems are thermodynamically open but closed in their organization. In other words, even though their material components turn over constantly, a material-independent property persists, which we call \textit{organization}. Moreover, organization comes from within organisms themselves, which requires us to explain how this \textit{self}-organization is established and maintained. In this paper we propose a mathematical and conceptual framework to understand the kinds of organized systems that living systems are, aiming to explain how self-organization emerges from more basic elemental processes. Additionally, we map our own notions to existing traditions in theoretical biology and philosophy, aiming to bring the main formal ideas into conceptual congruence.
\end{abstract}

\begin{keyword}
Organisms \sep autonomy \sep self-organization

\end{keyword}

\end{frontmatter}


\begin{quote}
``\textit{In vain we force the living into this or that one of our molds. All the molds crack. They are too narrow, above all too rigid, for what we try to put into them.}"
\newline
--Henri Bergson\cite{bergson1984creative}.  
\end{quote}

\begin{figure}[ht]
    \centering
    \includegraphics[width=0.75\textwidth]{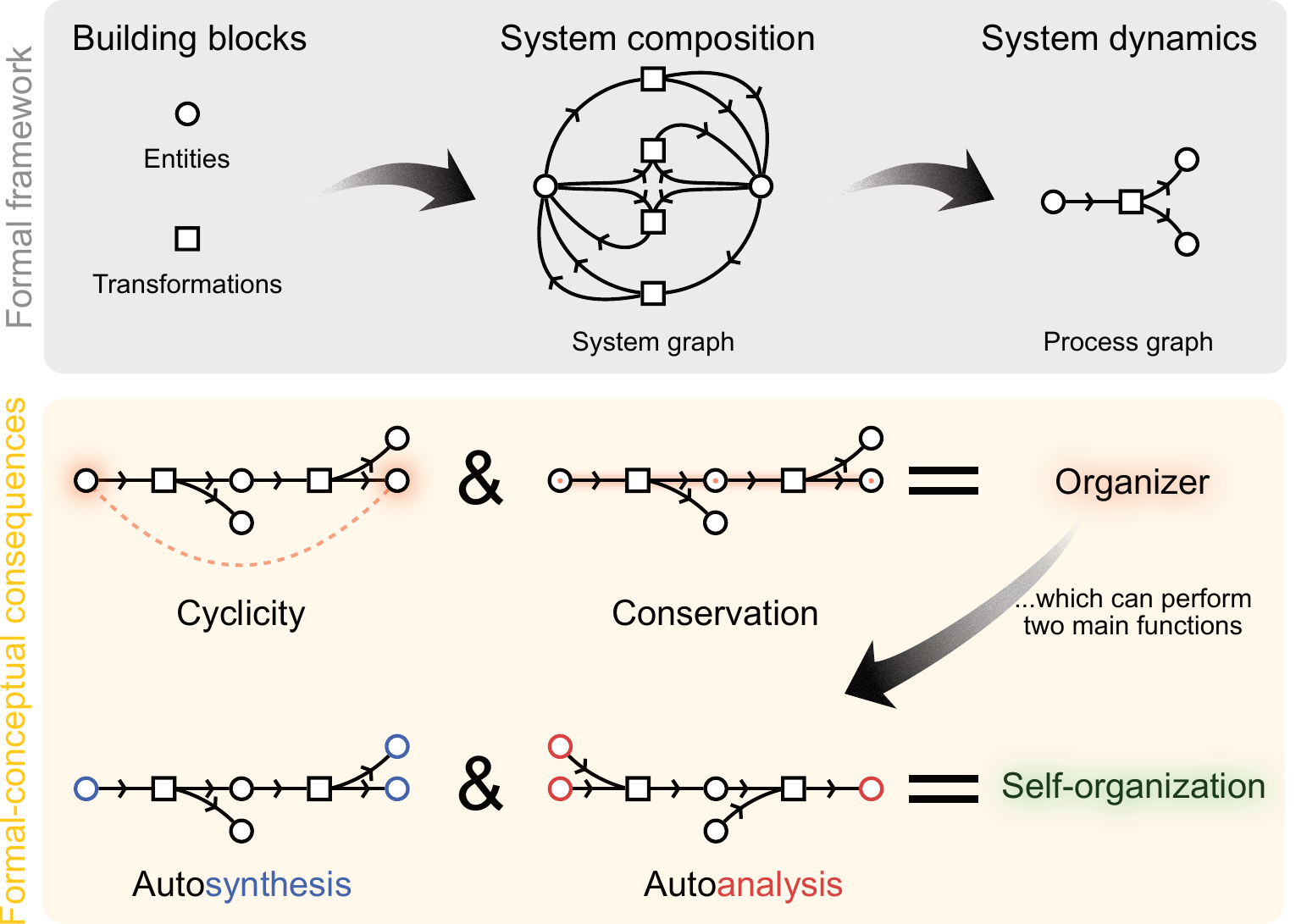}
    \label{fig:graphical_abstract}
\end{figure}

\tableofcontents

\section{Organization in living systems}
Living systems exhibit complex organization, with many interacting components producing diverse organized structures and behaviors. Organisms continuously transform matter and energy when coupled to their relevant environments. The complexity of these internal-external interactions often makes it hard to grasp what it is that makes a specific system organized, particularly when attempting to formalize this notion. Traditional theoretical approaches are often based on physical and thermodynamical constraints, frequently requiring system-specific details and assumptions which lose generality when applied to a large class of systems, like the class of the living. Here we take an alternative approach, with minimal physical assumptions or constraints, focusing instead on the features that make organisms \textit{organized}, which to us is the foundation of the living state. 

We develop a formalism that allows us to express what is \emph{possible} for a system, but we acknowledge that physical details and the context of particular systems will influence what is \textit{probable}. We aim to formally define system-level properties that living systems exhibit across the tree of life, which should also apply to a more general, non-Earth-bound, theory of the living. Our starting point is that despite the endless forms and behaviors exhibited by organisms in the biosphere, they are all organized in some equivalent sense. Thus, it is important to develop a formally grounded and conceptually coherent blueprint of this equivalence. 

\subsection{Approaches to the study of living organization}
The study of biological organization has a long history. However, despite this tradition, a formal approach has largely remained peripheral. Some efforts, related but not primarily focused on organisms, emerged in the context of chemical systems, like the work by Ostwald \cite{peng2022wilhelm}, Rössler \cite{rossler1971systemtheoretisches}, Gánti \cite{ganti1975organization}, Turing \cite{turing1990chemical}, chemical organization theory and autocatalytic systems theory \cite{dittrich2007chemical,heylighen2015chemical, kauffman1986autocatalytic, xavier2022small, hordijk2022autocatalytic}.

More general formalisms with a focus on organisms and organization have been scant, but we can mention organismic set theory by Rashevsky \cite{rashevsky1954topology,rashevsky1955some, rashevsky1955sometheorems}, the theory of autopoiesis by Maturana-Varela \cite{maturana1975organization, maturana1991autopoiesis}, Ehresmann-Vanbremeersch's theory of hierarchical evolutive systems \cite{ehresmann1987hierarchical}, and (M,R)-systems theory by Rosen \cite{rosen1958relational, rosen1958representation}. Similar in spirit but distinct in their details are the biotonic laws by Elsasser, catastrophe theory by René Thom, and Waddington's landscapes. To us, the work of Rashevsky-Rosen and Maturana-Varela form a particularly important core of ideas that strongly resemble our own motivations. An important distinction from this core, is that we consider \textit{consumptive} processes to be as important as \textit{productive} ones, to maintain self-organization. Most of the literature above ignores, or makes implicit, the consumptive aspects, which we will call \textit{autoanalysis}. Recently, a handful of theories describing particular aspects of living systems have been proposed, like the free energy principle \cite{karl2012free, kirchhoff2018markov}, assembly theory \cite{marshall2021identifying, sharma2023assembly}, constructor theory \cite{deutsch2013constructor, marletto2015constructor}, and others (e.g. \cite{england2013statistical, martyushev2006maximum}), but these are less relevant to our purposes.

Despite the variety of the aforementioned approaches, we agree with Mossio when he points out that there is still a ``blind spot for organization" in biology \cite{mossio2023introduction}. In our view, the blind spot lies in the fact that some abstract theories are hard to map to concrete systems, while others take the opposite stance by describing life as the mere ``jigglings and wigglings of atoms". We take neither of these approaches. To us, life is a phenomenon that \textit{lives} in the material world but whose defining properties are not tied to concrete material substrates. While this ontology of the living departs from mainstream theoretical biology, in our view it is not in conflict with it, but in complementarity.

\subsection{Lost and found: organisms as objects of study in biology}\label{sec:livingsystems}
\label{sec:sec1}
For much of modern science as it arose, organisms were the central objects of study in biology. However, during the last century, particularly after the emergence of the modern synthesis and the development of molecular biology, the organism took a progressively minor role. Over time, organisms became simple `vehicles' for genes and traits, which replaced them as the essential units of analysis and explanation. The modern synthesis and molecular biology are supra-organismal and infra-organismal disciplines, and thus a big part of the contemporary edifice of biology was built by excluding organisms. Moreover, genes were progressively endowed with qualities that used to be exclusive of organisms. Genes became the roots for any biological observable (phenotype), the causes for adaptive advantages/disadvantages, and they even gained agency and autonomy (e.g. in the `selfish gene' metaphor). Although there are some signs that the organism is returning to biology as a central object of study, much work remains to be done given the historical halt.

One of the first studies of what an organism is comes from Kant, who developed the idea of self-organized beings as those systems in which the constituent parts exist \textit{for} and \textit{by} the system \cite{kant1987critique}. This dual aspect of the components of an organism is important conceptually and formally, because it requires the elements of organisms to play at once both \textit{functional} and \textit{structural} roles. We will discuss this duality throughout our paper, but here we can mention that modern biology largely treats this duality as a hierarchy instead: structure \textit{implies} function. We believe that removing this hierarchy and recovering Kant's duality is necessary to understand why and how living systems are organized, and a primary aim of our work is to formalize these philosophical notions. Kauffman recently borrowed Kant's definition in his own work (calling such systems `Kantian wholes'), to build on a list of essential properties that living systems exhibit \cite{kauffman2024emergence, kauffman2013life, kauffman2014prolegomenon}. Although Kauffman has come close to a conceptual answer to what makes living systems alive, we believe that more work is needed to provide a formal grounding to some of these concepts.

The essence of our method is to define organic, self-organized systems in terms of the properties that \textit{emerge} from assembling simple processes that on their own are not self-organized. That is, instead of prestating the list of features we think organisms ought to have, we want to observe such features arise from these simpler processes. It is known that the kinds of molecules, genes, tissues, and other components vary widely across and within species, and even temporally within the same organism. To us, this strongly suggests that whatever life is, it will unlikely be found in a specific component, and more likely is a product of the \textit{relations} these components exhibit inside and outside the organism. Moreover, this also suggests that the question of what organisms are is primarily a theoretical question, given that no single \textit{real} organism can represent life as a phenomenon, just as no single molecule can encompass chemistry. This was the essential thrust of the research program developed by Rashevsky-Rosen, which was termed `relational biology' around the 1950s, and forms an essential pillar of our own research.

Our formalization follows a minimal set up, focused on how simple entities interact and change. Despite its simplicity, we observe emergent properties that we can map to essential features of living systems, in particular autonomy and self-organization. One important aspect that emerges, often overlooked in other formalisms, is that of the deliberate breaking down of components in the system, as opposed to simple entropic degradation or decay. This self-breakdown is balanced against the often-discussed self-building of organisms. We primarily leverage the theory of interacting networks \cite{meseguer1990petri, baez2018quantum, baez2021categories}, which has been successfully utilized to model diverse systems in engineering and the sciences \cite{murata1989petri, baez2019network}, but we also take inspiration from theories of chemical organization and autocatalysis \cite{dittrich2007chemical,barenholz2017design,blokhuis2020universal,deshpande2014autocatalysis}. As we introduce our formal tools, we will focus our discussion on the biological intuitions behind these, providing essential mathematical details whenever needed.

\section{A formal approach to the study of organisms}
We are interested in living systems where we can identify discrete entities and their interactions through whatever empirical or theoretical method at our disposal. Depending on the level of analysis and/or observation, these entities could be molecules, cells, tissues, multicellular organisms, and others. This means that, in our theory, there is no preferential level of description of a living system, so a `proper' description is that which suits our purposes of study. Therefore, we will begin with a minimal concept of entities, states, and transformations, and build on that throughout the paper to see how organization emerges.

\subsection{System entities and states}
Let us consider a finite set $E$ of \emph{entities}. Each entity lacks any further structure other than its kind. Later, we will consider subunits within entities, to show conservation properties in systems. For now, let us equip every entity in $E$ with a finite number of copies, or counts, and define the set of all entity counts as an \textit{ensemble}. In this way, we can think about the ensemble as the representation of the \textit{state} of the system at some given time. Thus, hereafter we will use ensemble and state as synonyms. Formally, each ensemble $x$ is a function $x:E\to\mathbb{N}$, mapping each entity to a natural number. We denote as $\mathbb{N}[E]$ the set of all possible states a system can have, where every ensemble $x$ in the set satisfies $\sum_{a\in E}x(a)<\infty$\footnote{This implies that only finite ensembles are states of a system, satisfied because $E$ is also finite. A distinction between state and ensemble is only relevant when $E$ is infinite, as it happens in combinatorial systems. For a more technical treatment of the relevant combinatorial system see \cite{ortiz2022combinatorics, ortiz2025combinatorial} by one of the authors.}. 

If $x$ and $y$ are ensembles, we define their sum, denoted as $x+y$, simply as the ensemble obtained by adding up the respective counts in $x$ and $y$. This means that for an entity $a\in E$, its count within $x+y$ is simply $x(a)+y(a)$, or $(x+y)(a)=x(a)+y(a)$. We denote the ensemble with zero counts simply as $\emptyset$. For ensembles $x$ and $y$, we say that $x$ is \emph{less than or equal} than $y$, denoted as $x\leq y$, if the counts in $x$ for each entity are less than those of $y$. Hence, $x\leq y$ whenever $x(a)\leq y(a)$, for each entity $a$. Similarly, $x\geq y$ whenever $x(a)\geq y(a)$ for each entity $a$. We write $x\ll y$ whenever $x(a)<y(a)$ for every entity $a$, and $x\gg y$ whenever $x(a)>y(a)$ for every entity $a$. Later, we will emphasize system properties that rely on this subtle distinction between $x\leq y$ and $<$.

\subsection{Elemental transformations}
The description of a system's state and composition is not enough to talk about organization. More important than this is to describe how these components interact and change to produce system-level features. A simple way to model change and interaction in our formalism is to allow transformations between different kinds of entities. Let us term these changes \textit{elemental transformations}, or transformations for simplicity. Let us denote this set of transformations with $T$. We can interpret transformations as changes in the properties, composition, or behavior of a given entity, so it becomes another kind of entity. This definition admits many kinds of transformations, whether they involve change in material composition or non-material, functional changes. For instance, a conformational change in an enzyme does not imply change in molecular composition but can involve functional change, whereas the breakdown of $ATP$ into $ADP + P_i$ does involve material change. Because this formalism admits representations at any level of biological organization, an entity can be a whole organism, which is useful if we wanted to model organism-level dynamics; for instance, to model when a chrysalis becomes a butterfly. As with any other modeling approach, it is our observations about a natural system of interest that inform the specification of our formalism.

\begin{definition}
    A \textbf{system} is a tuple $(E,T, \mathtt{in}, \mathtt{out})$ that consists of a pair of sets $E$ and $T$, as well as two functions $\mathtt{in}:T\to\mathbb{N}[E]$ and $\mathtt{out}:T\to\mathbb{N}[E]$ of inputs and outputs.
\end{definition}

Let $(E,T,\mathtt{in},\mathtt{out})$ be a system. Such system can be in a particular state (finite ensemble) at the beginning and can continue to change given the possible transformations in $T$. The inputs and outputs of transformations are dictated by the functions $\mathtt{in}$ and $\mathtt{out}$, which map input ensembles to output ensembles. For a transformation $\alpha$ where $\mathtt{in}(\alpha) = x$ and $\mathtt{out}(\alpha) = y$ we can equivalently write $x\xrightarrow{\alpha}y$. Let us here discuss a minimal example, to build up a more general intuition later. Consider a set of entities $E = \{a,b,c\}$ along with a set $T=\{\alpha,\beta\}$ of transformations, with inputs and outputs described as follows:

\[
a+b\xrightarrow{\alpha}c,\quad b+c\xrightarrow{\beta}a
\]

Now, let us suppose that the initial ensemble/state of the system is $3a+2b$. Then, we can see that in the next time step the only possible transformation would be $\alpha$, because the system does not have enough inputs for the $\beta$ transformation (it lacks $c$). Note that we use `time step' to denote the unfolding from inputs to outputs. A representation that naturally emerges when studying systems like this example is a directed graph, which features two types of nodes: one type represents the entities, while the other represents transformations. This means that every entity in \(E\) and each transformation in \(T\) has a corresponding node. There are directed edges from an entity \(a\) to a transformation \(\alpha\) in the quantity determined by \(\mathtt{in}(\alpha, a)\). Similarly, for the function \(\mathtt{out}(\alpha, a)\), we add the corresponding edges from \(\alpha\) to \(a\). Depending on the system of study, more than one edge can enter and leave a given pair of nodes, based on the counts of the functions $\mathtt{in}$ and $\mathtt{out}$, in which case we obtain a \emph{multigraph}. We refer to the structure that contains all of the properties described above as a \emph{system graph}. In essence, a system graph contains all the information needed to understand the behavior of the system it represents. We show an example in Figure~\ref{fig:systemgraphVector}.

\begin{figure}[ht]
    \centering
    \includegraphics[width=0.75\textwidth]{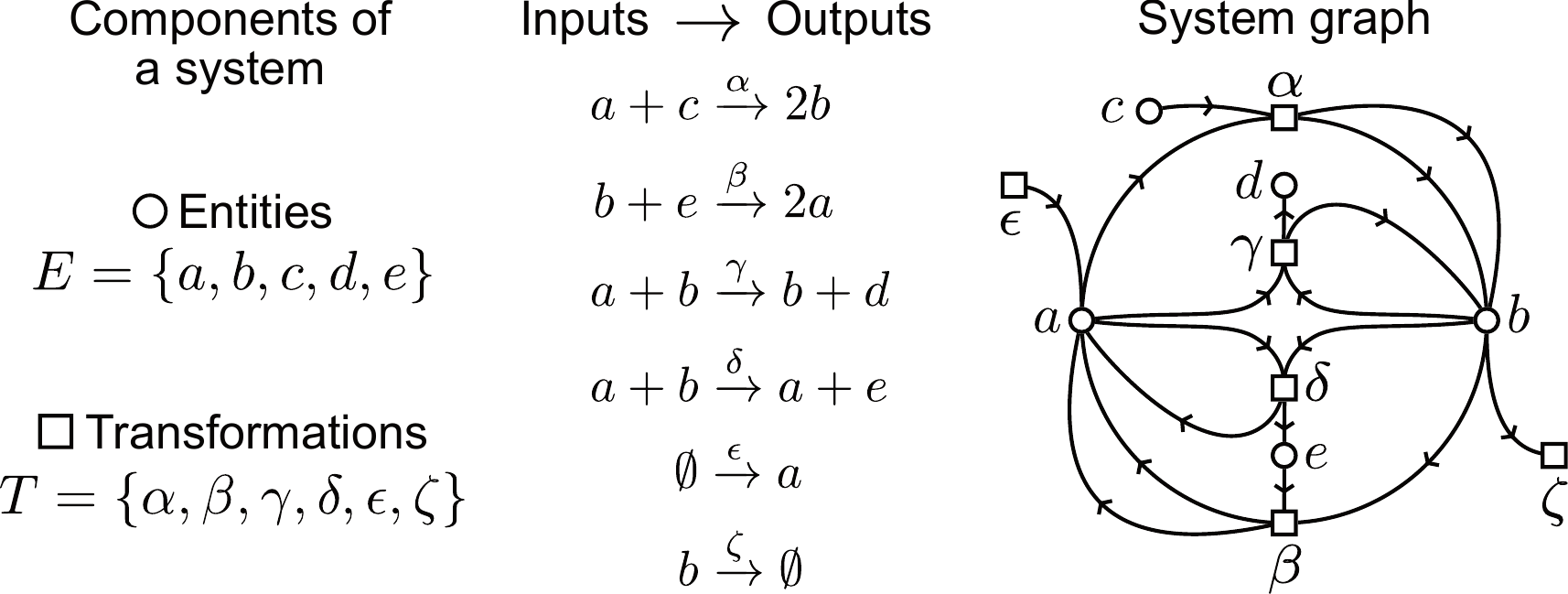}
    \caption{We can represent how entities ($E$) and transformations ($T$) interact via a system graph, shown above. Note that nodes are of two kinds, entities or transformations, and that multiple edges can enter and leave a given pair of nodes. Formally speaking this is a directed bipartite multigraph.}
    \label{fig:systemgraphVector}
\end{figure}

\section{General processes and temporal unfolding}
As we said before, system-level behaviors are more interesting to us than system components, as they reveal the dynamic aspects that underlie biological organization. Our starting point was to represent the structure of a system through its system graph. We might think of this structure as a blueprint of what a system can do, but in general, this does not say anything about what the system will actually do. This is a subtle point: even though the elemental transformations in $T$ already represent a minimal dynamical unit, from their inputs to their outputs, what we really want to describe is the general dynamical behavior of the \textit{whole system}. The dichotomy between the possible and the actual behaviors of a system, which in our formalism map to the system graph and the trajectory of system states, respectively, is strongly related to the genotype-phenotype dichotomy, but a detailed discussion of this will be left for future work.

\subsection{System dynamics}
In general, to observe the `temporal unfolding' of the system graph, we need to allow all possible interactions between the elemental transformations in $T$. These temporal interactions generate trajectories that can be thought of as the entire collection of dynamical behaviors that the system can exhibit. Importantly, we do not include a \textit{metric} notion of physical time, but a \textit{processual} one. That is, we do not define a specific timescale for any of the transformations in $T$, or for the interactions between these. This keeps the framework general, as to focus on which processes occur rather than focusing on how fast they occur.

The conceptual backdrop of focusing on processes is based on \textit{process biology} \cite{nicholson2018everything}. In this process ontology, the prime focus is on the `flow' of system states (phenotype/process) over time, instead of focusing on their persistent/static components (genotype/substance). Therefore, instead of giving structural details about the entities in $E$, we will focus on the changes these entities undergo within a system. A full treatment of the process ontology of the living is well beyond the aims of this paper, but ample literature exists on the topic (e.g. \cite{sep-process-philosophy, hertz2020nouns, dupre2014process,kearney2023process}).

\subsection{General processes as process graphs}
We define processes in terms of bipartite directed graphs. These \textit{process graphs} look similar to system graphs, but there are important differences. A process graph will have a `horizontal' layout, where we read from left to right the sequence of transformations that a system is undergoing. Importantly, there can be multiple copies of each node, both transformation and entity nodes, because over time a system can produce/apply entities/transformations of the same type. Process graphs are \textit{labeled} bipartite graphs, labeled by the elements of $E$ and $T$. In the rest of the text, we'll often use `process' and `process graph' interchangeably. The number of nodes with a given label represents the number of copies or instances of the corresponding entity or transformation in the process graph. Lastly, for a graph to qualify as a process graph it needs to satisfy a couple of topological conditions, making it clearly distinct from a system graph previously described:
\begin{enumerate}
    \item It should not possess any directed cycles.
    \item Each entity node can only have one outgoing and one incoming edge.
\end{enumerate}
These conditions are not artificially imposed, but arise naturally whenever a natural system unfolds in time. The first condition prevents `time loops', and the second ensures that the creation/destruction of entities is not acausal (i.e., every component is consumed/produced by a valid transformation). What is prevented by the first topological condition is that the \textit{literal} same entity appears at two or more distinct time steps, \textit{simultaneously}. In a system graph, loops are allowed because there is no notion of physical time, so there is no `ordering' of events: everything is shown `at once'. Note that the first condition does not prevent \textit{kinds} of entities to appear iteratively through time sequences in process graphs. In fact, a general notion of cyclicity will be fundamental to our definitions of self-organization, and relies on processes that produce and re-produce entities of the same kind over time. Re-creation is an essential feature of living systems: even though their \textit{literal} material components turn over constantly, many \textit{kinds} of entities/components re-appear over and over in `organismal time'.

\subsection{Process graph composition}
Any terminal node in a process graph can be extended by matching and merging new nodes that preserve the topological validity we just described. This is equivalent to saying that, as long as there are valid ways to add new nodes/edges, the process graph can keep unfolding in time. In the definition below we use the term \emph{matching} when referring to any graph obtained by performing such an extension. Thus, we say that a match in a graph is \emph{valid} if the resulting matching (graph) is also valid. 

\begin{definition}\label{def:process}
    A \textbf{process} is a valid matching of a set of disconnected copies of elementary process graphs. 
\end{definition}

Let us now describe more precisely how to build general process graphs from individual transformations or from pre-existing process graphs. First, we note that each elementary process $\alpha\in T$ is itself a graph that includes the transformation node $\alpha$ and its corresponding input and output nodes, given by $\mathtt{in}(\alpha,a)$ and $\mathtt{out}(\alpha,a)$. Then, it is easily noted that these elemental processes can be combined in many ways, as long as they respect the topological validity we described above. For completeness, we also include single entity nodes (with no outgoing edges), which we term null processes, to serve as elemental steps where `nothing happens'. Then, given a process $p$, we can refer to the ensembles of entity nodes that do not have an incoming edge or an outgoing edge as \emph{inputs} and \emph{outputs} of the process, respectively. Hence, similar to an elementary transformation, any process $p$ has input and output nodes, but these correspond to the states of the \textit{entire system}, given by the input $x$ ensemble and output $y$ ensemble. We can succinctly denote such a process as $x\xrightarrow{p}y$. If we have some other process $x'\xrightarrow{q}y'$, we refer to the disconnected graph of $p$ and $q$ as their \emph{parallel composition}, denoted by $p+q$. We can also build a \emph{sequential composition} of processes $x\xrightarrow{p}y$ and $y\xrightarrow{q}z$, simply because the output of the first is exactly the input of the second. It is important to note here that there is not one but potentially many sequential compositions of $p$ and $q$, given that there may be more than one valid matching, arising from the permutations of $y$. Hence, in general, the notation $pq$ for sequential composition is not well-behaved. Thus, in this work, we will be working with general matchings that obey definition~\ref{def:process}, which include cases where the outputs of a process coincide only \textit{partially} with the inputs of the next process\footnote{For a more formal treatment of the kinds of parallel and sequential compositions see \cite{baez2021categories}.}. Biologically speaking, this feature, of having multiple potential valid $pq$ compositions is at the center of self-regulation and historical contingency, since biological processes are not deterministic and the kinds of behaviors an organism exhibits are not `uniquely' defined. In essence, the fact that composition is not well-behaved gives us, for free, a mechanism to create dynamic, non-deterministic, behavioral diversity. We show all composition mechanisms in Figure~\ref{fig:processes}.

\begin{figure}[ht]\label{fig:processes}
    \centering
    \includegraphics[width=1\textwidth]{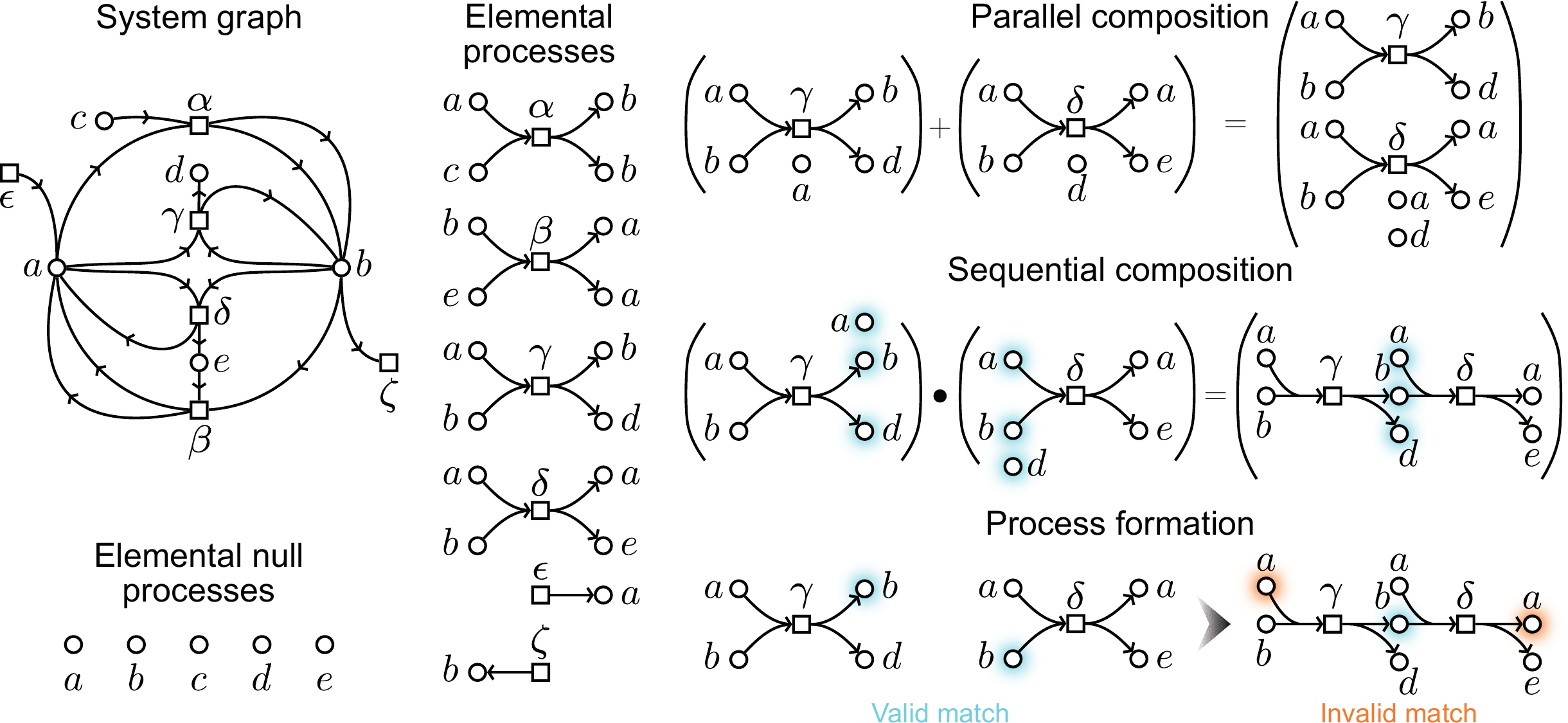}
    \caption{A system graph (left) is the composite of a collection of elemental processes, which include null processes. Elemental processes can be composed in parallel or sequentially, preserving the topological validity described in the text. Note that a process graph has multiple copies for some nodes, and no directed loops or multiple edges, while the system graph contains unique nodes, directed loops and multiple edges.}
\end{figure}

Matching elementary processes as we described allows us to `unfold' the system graph, which is inherently static, into a dynamic process graph that can grow as long as any valid matching is possible. This dynamic unfolding of processes impose a pattern of connectivity between all of the states that a system can exhibit. The resulting topology of states can form the basis of a \textit{process pathscape}, describing the collection of behavioral paths that a system can explore from a given state, akin to a phase portrait in dynamical systems theory. The important difference is that there are multiple possible trajectories taken from the same starting point, and there is no principled way to know which one will actually be taken. Another important point is that the issue of connectivity within a single process represents the causality and dependence between its constituting sub-processes. This will be very useful when building up notions of autonomy and organizational closure. Let us now introduce some useful notions from linear algebra, which will help us state some of our formal results more precisely and succinctly, though we will map our intuitions back to the system and process graphs we just described.

\subsection{System matrices and subsystems}

In a small system graph, it is often easy to detect cycles by inspection of its graphical representation. However, in arbitrarily complex systems like biological and ecological ones, visual inspection is no longer possible or effective, and other procedures become necessary to determine system-level properties. For this reason, here we introduce some basic notions of linear algebra that will aid our analysis. This will also allow us to talk about subsystems in a simple way, which will be important in order to emphasize what part of a system has a certain property that might not be exhibited by the whole system. Moreover, the use of matrix notation simplifies the discussion of key concepts such as conserved quantities and cycles, which will be central to our formal definitions. We first note that any system in the way we have defined them can be equivalently represented by a pair of matrices, without loss of information.

Let us consider a system $(E,T,\mathtt{in},\mathtt{out})$ as defined in Section~\ref{sec:livingsystems}. From this, we construct a pair of \emph{system matrices}, $\mathcal{S}=(\mathcal{S}_0,\mathcal{S}_1)$, corresponding to the input and output functions, respectively:
\[
\mathcal{S}_0^{ij}=\mathtt{in}(i,j),\qquad\mathcal{S}_1^{ij}=\mathtt{out}(i,j),\qquad \text{for } i\in T, j\in E.
\]
The information contained within $\mathcal{S}$ is sufficient to fully specify the system, hence we will often use $\mathcal{S}$ as a proxy for discussing the system itself. Because $\mathcal{S}$ is the matrix representation of a system, subsystems can be easily defined as matrices that are component-wise smaller than the system matrix. Formally, a subsystem $\mathcal{A}$ of $\mathcal{S}$ is denoted by $\mathcal{A}\leq\mathcal{S}$, where for each transformation $i$ and entity $j$ we have
\[
\mathcal{A}_0^{ij}\leq\mathcal{S}_0^{ij},\text{ and }\mathcal{A}_1^{ij}\leq\mathcal{S}_1^{ij}.
\]
When we use $\mathcal{A}<\mathcal{S}$ we indicate that $\mathcal{A}$ is a \textit{proper} subsystem of $\mathcal{S}$, meaning that it is strictly smaller than $\mathcal{S}$ in at least one transformation/entity entry. We then can represent states as vectors rather than functions, so for an ensemble $x\in\mathbb{N}[E]$, we write $x_i$, instead of $x(i)$, for the entry corresponding to entity $i\in E$. Apart from the symbolic change, this vector notation enables us to encode more general entities than discrete ensembles, such as vectors with negative integers or real numbers, and vectors over the set of transformations, which can be useful in many different contexts, although in this work we will keep working with discrete entities. 

\begin{figure}
    \centering
    \includegraphics[width=0.8\textwidth]{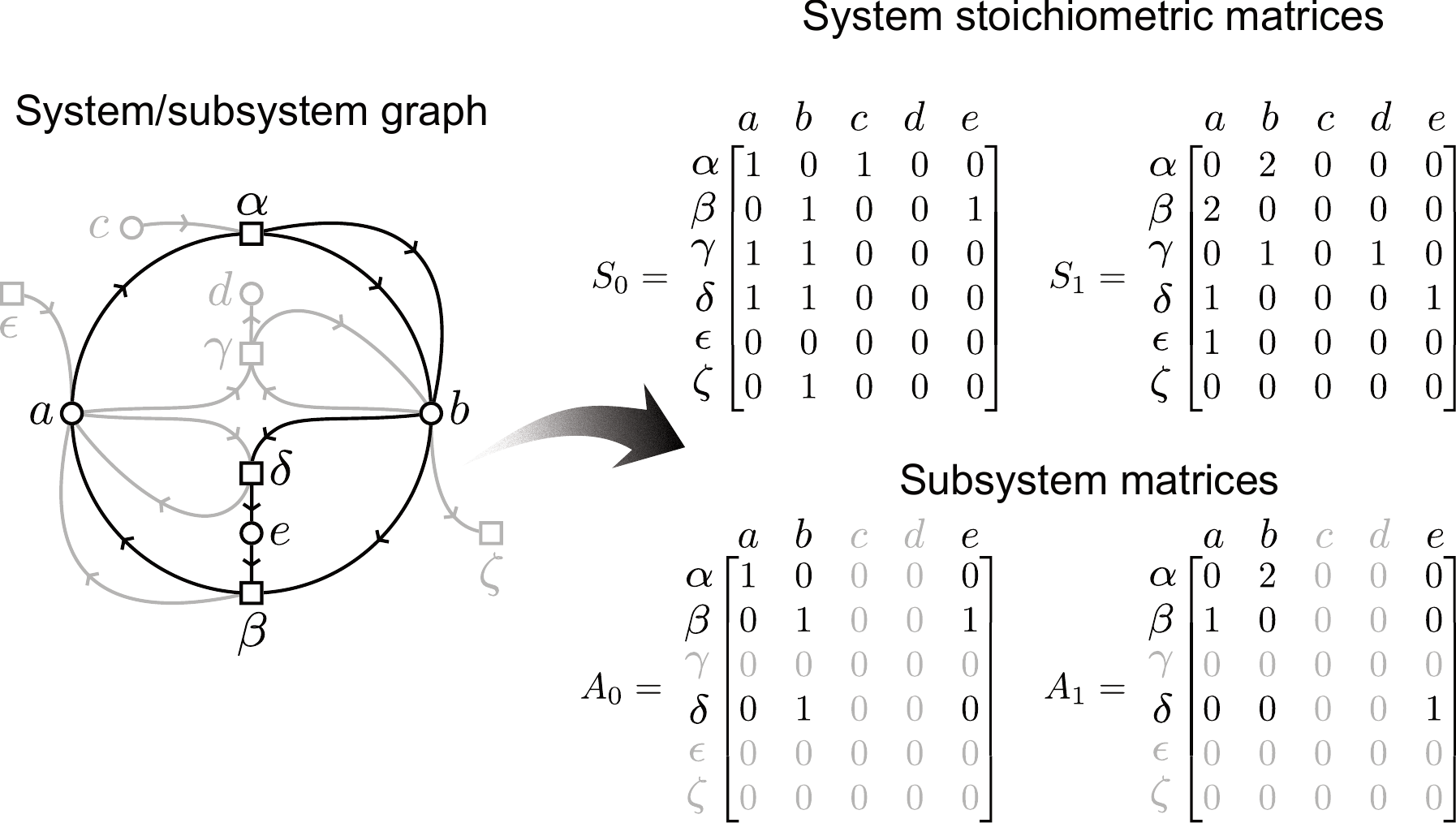}
    \caption{A system graph can be partitioned into subsystems (left), and this information can be equivalently represented as a pair of matrices (right), as explained in the text. Black represents the subsystem in the graph and subsystem matrices, and gray the components outside the subsystem.}
    \label{fig:matrixes}
\end{figure}

One of the main advantages of using matrix and vector representations is the possibility of using matrix multiplication. Let us recall that for an arbitrary matrix $\mathcal{M}$ of dimension $m \times n$ and a column vector $x$ of dimension $n$, the multiplication $\mathcal{M}x$ results in a vector of dimension $m$:
\[
(\mathcal{M}x)_j=\sum_{i}\mathcal{M}^{ij}x_i.
\]
For a vector $q$ representing the quantities/counts associated with each entity in the system, the vector $(\mathcal{S}_1-\mathcal{S}_0)q$ indicates the net change of $q$ through each transformation. Similarly, if $u$ is a vector over transformations, the multiplication $(\mathcal{S}_{1}^\top -\mathcal{S}_{0}^\top )u$ yields a vector detailing the net change in each entity when such transformations are applied. Here, the superscript $^\top $ denotes a transpose of a matrix, switching the matrix's rows and columns. For simplicity, hereafter we will denote the difference between the system matrices as $\mathcal{S} = \mathcal{S}_1 - \mathcal{S}_0$. Note that this is the same symbol as the matrix representation $(\mathcal{S}_0,\mathcal{S}_1)$, which is not the same as $\mathcal{S}_1 - \mathcal{S}_0$, but the specific use should be clear within a context.

Given a process $p$, we will often need to speak about the \emph{transformation vector} associated with such a process. This is a vector $u\in\mathbb{N}^T $ for which each entry describes the number of times the corresponding transformation appears within the process. A process always has a unique transformation vector. The converse, however is not true. In most cases, given a transformation vector $u$, there will be multiple processes yielding the same transformation vector, for similar reasons that we explained when noticing how a system graph can generate many different process graphs.

\section{Two formal properties leading to organization: conservation and cyclicity}

In this section, we introduce the concept of conservation and cyclicity within a system, which are intimately related. To show this relationship, let us discuss the notion of a catalyst, which we can think of as an entity (or entities) that are preserved across a process, where the other entities in the process have changed. This is an emergent property resulting from the composition of processes, which in some situations can form \emph{catalytic cycles}. It might seem intuitive to define catalytic cycles as some process $x\xrightarrow{p}y$ for which $x$ and $y$ share at least one entity, the catalyst. However, we would mistakenly classify the process consisting of the parallel composition of $x\xrightarrow{p}y$ with $z\xrightarrow{q}x$ as catalytic. However, the ensemble $x$, seemingly conserved in $p+q$, is actually consumed in $p$ and subsequently regenerated in $q$, which was clearly not preserved in the process, as catalysts are. Therefore, we need a more nuanced notion of catalysts. Below we provide a definition in terms of system and process graphs, but later we will recast our formulations in the terms of matrices.

\begin{definition}
A \textbf{catalytic cycle} is a directed cycle in the system graph. A \textbf{catalytic process} is a process for which a path exists from an input to an output entity with the same label.
\end{definition}

Let us note that a single catalytic cycle, in the system graph, can give place to many different catalytic processes. Conversely, a catalytic process can embody multiple catalytic cycles. This is because of the multiple permutations that a valid matching can have, as explained before. Regardless of this, a true catalyst should have some conservation property preventing us from considering materially disjoint entities as catalysts, as discussed previously. In other words, we should be able to trace the identity of a catalyst throughout a connected path in some process for it to be considered a true catalyst, as shown in Figure~\ref{fig:catalysis}. Despite the terminology which we borrowed from chemistry, let us note that we are not making any assumptions about the material substrates which embody any of the properties of our formalism, so these can be observed at any level of biological and physical organization, as we have discussed before. 

\begin{figure}
    \centering
    \includegraphics[width=0.6\textwidth]{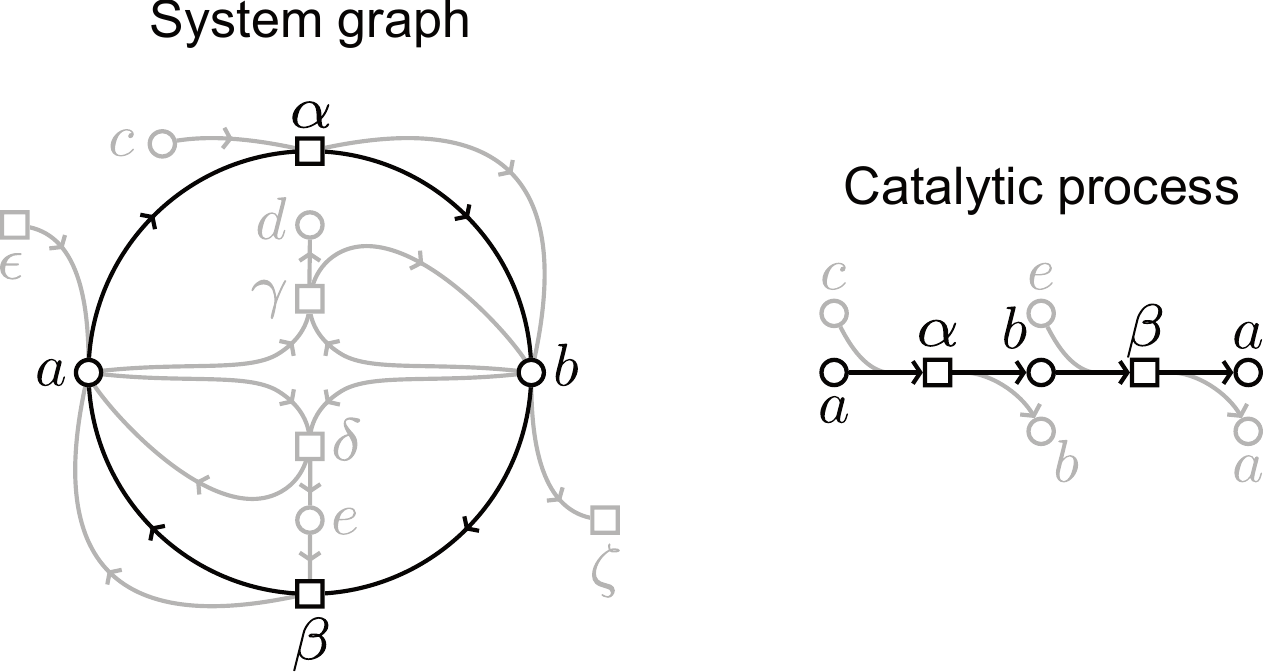}
    \caption{We show a catalytic cycle in the system graph (left, in black) along with a catalytic process graph (right), with node \textit{a} as the catalyst in this case. In this case the catalytic process is very similar to the catalytic cycle, but as we explain in the text, these are usually not one-to-one correspondences, given that a general process can take many alternative paths.}
    \label{fig:catalysis}
\end{figure}

By the definition~\ref{def:process}, the graph of a process is absent of directed cycles. Notice that a pair of input and output entities with the same label is invalid as it results in a directed cycle. Therefore, a process has a catalytic cycle whenever an invalid match exists.

\subsection{Conservation}
In order to find a true catalytic process, we need to be able to trace the \textit{conservation} of a catalyst throughout the process. Here, we do this by tracing \textit{subunits} within each entity in a process, as to trace what is preserved when each transformation `unfolds'. This is akin to labeling the material constituents of each entity, so that we can follow these labels over time and identify which of these persist through a process. This method resembles what is done in biochemistry, for instance, where the order of events and mass conservation in a pathway are elucidated by radioactive labeling of specific atoms in the participant molecules.

Let us represent the subunit structure of every entity in a system as a vector $n\in \mathbb{N}^E$, which for each entity $j$ it gives the number $n_j$ of subunits contained within it. In essence, $n$ encodes the cardinalities of sets of subunits, which, as noted earlier, can be thought of as `tagged' indivisible building blocks of an entity, but that can become part of distinct entities through every process. Equipped with this notion, we can now say that $n$ is a \emph{conserved} quantity if $n>0$ and the following conservation equation is satisfied:
\[
\sum_{j}\mathcal{S}_0^{ij}n_j = \sum_{j}\mathcal{S}_1^{ij}n_j,\quad\text{for all }i\in T.
\]
It is important that the vector $n$ be non-zero, as the zero vector trivially satisfies the conservation equation. We may describe a conservative system with $n\gg 0$ as \emph{strongly conservative}. Typically, when dealing with entities of non-negligible mass that adhere to mass conservation principles, the vector that represents the mass of each entity qualifies the system as strongly conservative.

In open systems that interact with others and with the environment, a weaker notion of conservation is more relevant. This is because in these systems we often find `cores' exhibiting strong conservation, while the rest of the system is not necessarily conservative. To make this concrete, let $0<\mathcal{A}\leq\mathcal{S}$ be a non-zero subsystem. We will refer to entities that participate in $\mathcal{A}$, either as inputs or outputs, as \emph{internal}. Otherwise, an entity is \emph{external}. A vector $n$ which is zero for all external entities is similarly an \emph{internal vector}. If the entries of an internal vector $n$ are positive for all internal entities, we say the the vector is \emph{full}.

\begin{definition}
    A \textbf{semi-conservative} system is a system $\mathcal{S}$ with a non-zero subsystem $0 < \mathcal{A} \leq \mathcal{S}$, and a full internal vector $n$, such that $\mathcal{A}n = 0$.
\end{definition}

\begin{figure}
    \centering
    \includegraphics[width=0.8\textwidth]{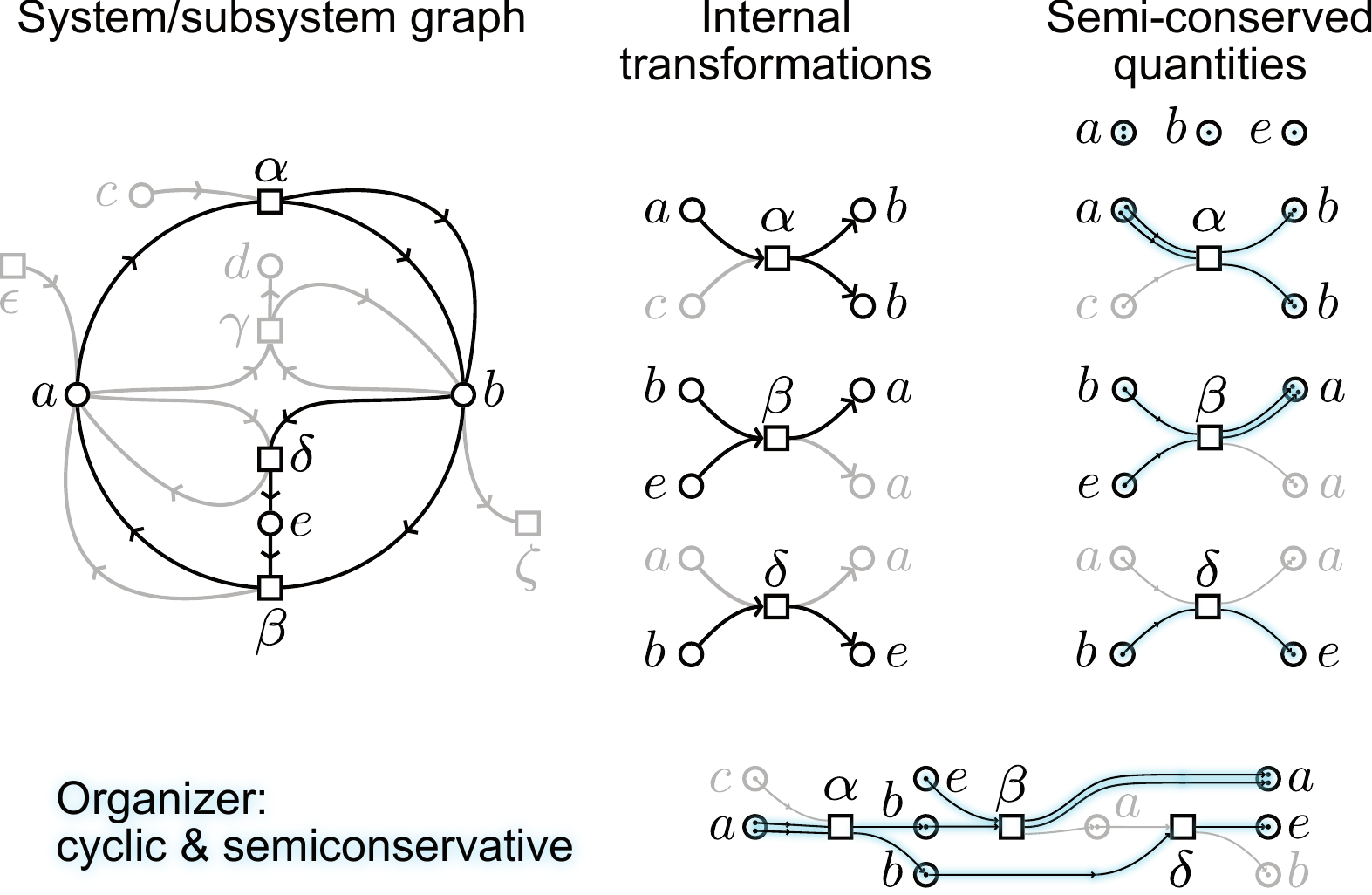}
    \caption{For a system/subsystem graph (left), we can `trace' the changes in the internal transformations (center) by using a notion of \textit{subunits} (right) that keep track of what is being transformed. In the right diagrams, we represent each entity as a circle with internal dots representing its subunits. The blue lines show the trace of subunits through the transformations. We show an organizer (bottom) as a cyclic and semi-conservative process.}
    \label{fig:conservation}
\end{figure}

\begin{proposition}\label{thm:84293018}
    Let $(\mathcal{S},\mathcal{A},n)$ be a semi-conservative system. Then, a transformation has an input in $\mathcal{A}$ if and only if it has an output in $\mathcal{A}$.
\end{proposition}
\begin{proof}
    Consider a semi-conservative system as in the hypothesis of the proposition. Suppose some transformation $i$ has an input in $\mathcal{A}$, meaning that $\mathcal{A}_0^{ij}>0$ for some entity $j$. Since $n$ is a full internal vector, this means that $(\mathcal{A}_0n)_i>0$. Moreover, since we require that $\mathcal{A}_0n=\mathcal{A}_1n$, we must have that $(\mathcal{A}_1n)_i>0$ as well, and so $i$ must have some output in $\mathcal{A}$. The converse is analogous.
\end{proof}

This property is referred to as `autonomy' in \cite{blokhuis2020universal}, and in \cite{barenholz2017design} it is part of the definition of an autocatalytic cycle, but in our paper, it only takes us partially toward a full notion of autonomy. 

\subsection{Cyclicity}
Let us come back to cycles, which we briefly discussed before in terms of system and process graphs, this time in terms of system matrices. Let us refer to a process of the form $x\xrightarrow{p}x$, where the input is equal to the output, as a \emph{cyclical process}. In terms of the transformation vector $u$ associated with such a process, one observes that the net change in the quantity of each species is zero. Formally, we have that the following cycle equation is satisfied:
\[
\sum_{i}\mathcal{S}_0^{ij}u_i=\sum_{i}\mathcal{S}_1^{ij}u_i,\quad\text{for all }j\in E.
\]
We deem a system with one such vector $u$ as \emph{cyclical} and we refer to $u$ as a \emph{cycle vector}. As before, it is important to require that a cycle vector be properly greater than 0, as otherwise any system would be trivially cyclical. For reasons we explained before, we will be interested in scenarios where only a subsystem $\mathcal{A}$ satisfies the cycle equation. In analogy with semi-conservative systems, we will refer to transformations as \emph{internal} if they have either an input or an output in $\mathcal{A}$. A transformation vector $u$ will be \emph{internal} if it is zero for all external entities, and an internal vector will be \emph{full} when it is positive for all internal entities. 

\begin{definition}
    A \textbf{semi-cyclical} system is a system $\mathcal{S}$ with a subsystem $0<\mathcal{A}\leq\mathcal{S}$ and a full internal vector $u$ such that $\mathcal{A}^\top  u=0$.
\end{definition}

Thus, a proposition analogous to that of semi-conservative systems exists:

\begin{proposition}\label{thm:89741902}
    Let $(\mathcal{S},\mathcal{A},u)$ be a semi-cyclical system. Then, an entity is an input in $\mathcal{A}$ if and only if it is an output in $\mathcal{A}$.
\end{proposition}
\begin{proof}
    The proof is analogous to that of Proposition~\ref{thm:84293018}.
\end{proof}

The property implied by this theorem is referred to as `semi-self-maintenance' in \cite{dittrich2007chemical} and it is also part of the definition of an autocatalytic cycle in \cite{barenholz2017design}, but we believe that in our approach these properties emerge from more fundamental processes, rather than being already part of a definition.

\section{Autonomy as the confluence of cyclicity and conservation}
\label{sec:autonomy}

We can now show how these system-level properties described above assemble into a behavior that we will name \textit{autonomy}, which we formalize through the notion of an \textit{organizer}. Thus, for us an \emph{autonomous} system is one that has the properties of an organizer.

\begin{definition}
    An \textbf{organizer} is a system $\mathcal{S}$ together with a non-zero subsystem $0<\mathcal{A}\leq\mathcal{S}$, a full internal entity vector $n$, and and a full internal transformation vector $u$, such that $\mathcal{A}n=0$ and $\mathcal{A}^\top u=0$.
\end{definition}

From this definition we can clearly see that organizers are both conservative and cyclical, but the converse is not necessarily true. It may be possible that a system $\mathcal{S}$ is both cyclical and conservative but via different subsystems. Thus, autonomy requires a system to be cyclical and conservative \textit{via the same subsystem}, that is, the organizer. In other words, autonomy is a stronger property than the mere conjunction of conservativeness and cyclycity. We will now use propositions~\ref{thm:84293018}~\&~\ref{thm:89741902} to prove a result that relates organizers to catalysis. This helps us building an intuition of why an organizer can be thought of as an autonomous system constituted by catalysts.

\begin{lemma}\label{thm:30925853}
    Each organizer contains at least one catalytic cycle.
\end{lemma}
\begin{proof}
    Let $(\mathcal{S},\mathcal{A},n,u)$ be an organizer. Let $a$ be any internal entity. By Lemma~\ref{thm:89741902}, $a$ must be the input of some internal transformation. Then, by Lemma~\ref{thm:84293018}, that transformation must have an output that is an internal entity. We continue in this fashion constructing a sequence of connected internal entities and transformations. Since there are only a finite number of internal entities and transformations, this procedure can only proceed for a finite number of steps before arriving at an entity or a transformation that we have already used before, thus forming a catalytic cycle.
\end{proof}

However, if a subsystem contains a catalytic cycle, this does not automatically make it autonomous. For that, we would need every entity to act as a catalyst, and every transformation to be involved in some catalytic cycle, as we just explained. In fact, we can prove something even stronger, which again illustrates the intimate connection between organizers and catalysis.

\begin{theorem}\label{thm:39572809}
    Let $(\mathcal{S},\mathcal{A},n,u)$ be an organizer. Then, there exists a process $p$ with transformation vector $u$ such that all of its inputs and outputs are catalysts.
\end{theorem}
\begin{proof}
    Consider the graph of the process given by the parallel composition of all the transformations in $u$. We will proceed by matching inputs and outputs, making sure not to introduce any directed cycles as to preserve the topological validity we described before. Given that there are only a finite number of inputs and outputs, this procedure is guaranteed to terminate. At that point, all possible matchings between inputs and outputs will yield a directed cycle, and hence a catalyitic cycle. What this means is that each entity appearing as an input (or output) is, in fact, a catalyst.
\end{proof}

Notice that since the vector $n$ is full, all internal entities will be part of some catalytic cycle in the process $p$ implied by the theorem, even if they do not appear as its inputs or outputs. Furthermore, since $u$ is full, all internal transformations are involved in some catalytic cycle. This is the sense in which an organizer can be considered autonomous: \emph{it is entirely constituted by catalytic cycles}. 

Let us now highlight the importance of catalysis in relation to autonomy, by considering minimal organizers. First, given an organizer $(\mathcal{S},\mathcal{A},n,u)$ we say that it is \emph{minimal} if no proper subsystem is itself an organizer. Formally, this means that for each subsystem $0<\mathcal{A}'<\mathcal{A}$, and subvectors $0<n'<n$ and $0<u'<u$, the tuple $(\mathcal{S},\mathcal{A}',n',u')$ is not an organizer. Thus, a cycle in the system graph is \emph{irreducible} if it cannot be expressed as the concatenation of two other cycles.

\begin{theorem}
    Minimal organizers are equivalent to irreducible cycles of the system graph.
\end{theorem}
\begin{proof}
    First, we establish that an irreducible cycle in the system graph is a minimal autonomous subsystem. Construct a vector $n$ for which each entity in the cycle has a value of $1$. Then construct a vector $u$ which has a value of $1$ for each transformation appearing in the cycle. Finally, construct a subsystem $\mathcal{A}$ that has a value of $1$ for each edge in the cycle. These data satisfy the requirements of an organizer. Suppose now we have a minimal organizer. By Lemma~\ref{thm:30925853}, such a system contains a catalytic cycle. That catalytic cycle is a cycle in the system graph, which, if it is not itself an irreducible cycle, it contains one. By the previous arguments, we can construct an organizer based on that cycle. But this means that our original system contains an organizer, and, since that original system was minimal, it means that it was an irreducible cycle to begin with.
\end{proof}

\section{Approaching life: dual aspects of organization}
In the rest of the paper, we will delve into various aspects of organization emerging from our formalism. Our focus will be on dual concepts of self-organization, which we term \emph{autosynthetic} and \emph{autoanalytic}. The bulk of literature on organization and abiogenesis directs its attention towards what we term autosynthetic organization: networks of entities and transformations that ensure the self-sufficient \textit{production} of the system's components. This emphasis proves particularly relevant in the study of life's origins, given the initial absence of living or self-organized entities, which necessarily requires proliferation or production in the first place. However, an equally important form of organization is consumptive, not productive: molecules undergo controlled degradation in cells, cell death in tissues is strongly regulated, and even individual organisms in populations die in a non-random way. This alternate perspective, of holding consumptive and productive processes as equals, is significant in scenarios where organizational components are abundant, and passive degradation or decay mechanisms are not sufficient to maintain control in the system. In these instances, a need arises for self-directed breakdown of components. To us, this is the essence of constructive-combinatorial systems such as chemical, biological, and ecological systems, where the participant components are produced, reused, and degraded in a \textit{self}-controlled fashion. Moreover, the breakdown of aged or dysfunctional entities becomes necessary to reuse basic building components and energy for other processes --- what is known as catabolism, for which anabolism serves a dual role.

\subsection{Autosynthesis}

In this section we focus on systems that are able to collectively increase their entities' populations, a property we term \emph{autosynthesis}. Here, given that \textit{internal} entities and transformations are fundamental, we assume that these are in play unless we write `external' explicitly. Formally, this means that $\mathcal{S}$ is now a system without external entities and transformations with respect to a subsystem $\mathcal{A}$ with an organizer.

\begin{definition}
    An \textbf{autosynthetic} system is a system $\mathcal{S}$ together with a non-zero subsystem $0<\mathcal{A}\leq\mathcal{S}$, full vectors $n$ and $u$ such that $\mathcal{A}n=0$, $\mathcal{A}^\top u=0$, and $(\mathcal{S}^\top u)_a>0$ for all entities $a$. 
\end{definition}

These conditions for autosynthesis are a special case of the `self-maintenance' property in chemical organization theory \cite{dittrich2007chemical}. Moreover, it is also related to the definition of `autocatalytic metabolic cycle' in \cite{barenholz2017design}, `autocatalysis' in \cite{blokhuis2020universal}, and `self-replicable' in \cite{deshpande2014autocatalysis}. It is also related to the `self-maintenance' property in \cite{dittrich2007chemical}. It is also somewhat connected to the `reflexively autocatalytic' property in \cite{hordijk2022autocatalytic}. In essence, our definition implies that this organized system can be understood as a network of catalytic cycles, all of whose constraints are produced within the system. This is closely related to the `closure to efficient causation' property postulated by Rosen \cite{rosen1991life}.

Notice that our definition of autosynthesis is given in terms of the vectors $n$ and $u$, which does not guarantee that there is a process which we could understand as autonomously synthetic. So, to ensure that this is in fact the case, we can refer to a process as \emph{autosynthetic} only when the following conditions are satisfied: 
\begin{enumerate}
    \item Every input has an invalid match. 
    \item The counts of each entity increase across the process. 
    \item Each output has a path from some input.
\end{enumerate} 
If every match is invalid then all matches generate some catalytic cycle. Moreover, since the counts of all entities increase, each input will have matches, and hence be a catalyst. Finally, each output is produced catalytically since there is a path to each output to some input, which is a catalyst. In other words, an authosynthetic process is autonomous and \textit{fully} synthetic. Note that the properties of autosynthetic processes are all preserved whenever multiple autosynthetic processes are composed in parallel. Furthermore, matching autosynthetic processes also preserves those properties. Then it makes sense to consider minimal, generating autosynthetic processes that cannot be decomposed into smaller ones by matching or parallel composition. 

Now we can show that despite being formulated in terms of vectors, the definition of autosynthetic systems implies the existence of a corresponding autosynthetic process graph. The use of vectors was a strategic step, as it is much easier to work with vectors than with process graphs, given that the latter have much more structure, but the fact that we can ensure the existence of this corresponding graph is important as to relate `formal' properties to `physical' ones (which occur over time). The following theorem shows this implication. 

\begin{theorem}\label{thm:85743890}
    Each autosynthetic system has a compatible autosynthetic process.
\end{theorem}
\begin{proof}
    Let $(\mathcal{S},\mathcal{A},n,u)$ be an autosynthetic system. We proceed as in the proof of Theorem~\ref{thm:39572809}, with the parallel composition of the transformations in $u$. We continue by matching inputs and outputs, ensuring that no directed cycles are formed. At some finite time point, all possible matches left will be invalid, and hence part of a catalytic cycle. Moreover, since $u$ is a transformation vector through which all entities increase, each input appears in excess within the outputs. Since all transformations have an input by Proposition~\ref{thm:84293018}, then the resulting process will have at least one input. Furthermore, all output entities will be connected to some input entity. Thus, the process is autosynthetic. 
\end{proof}

\subsection{Autoanalysis}

For us, \textit{autoanalysis} is the dual property of autosynthesis, and can be defined as the ability to autonomously consume the entities involved in the system. Because autoanalysis and autosynthesis are dual aspects of organization, the discussion here will be shorter, as it largely mirrors the preceding text.

\begin{definition}
    An \textbf{autoanalytic} system is a system $\mathcal{S}$ together with a non-zero subsystem $0<\mathcal{A}\leq\mathcal{S}$, full vectors $n$ and $u$, such that $\mathcal{A}n=0$, $\mathcal{A}^\top u=0$, and $(\mathcal{S}^\top u)_a<0$ for all entities $a$. 
\end{definition}
 
We define \emph{autosynthetic processes} in analogy to the previous section. In short, an autoanalytic process is one in which:
\begin{enumerate}
    \item Each output has an invalid match.
    \item Every entity decreases its count. 
    \item Each input has a path to some output.
\end{enumerate}

\begin{theorem}
    Each autoanalytic system has a compatible autoanalytic process.
\end{theorem}
\begin{proof}
    The proof is analogous to that of Theorem~\ref{thm:85743890}.
\end{proof}

\section{\textit{Ōmeteōtl}: Self-organized systems and their resilience}

As we discussed in Section~\ref{sec:autonomy}, we can understand autonomy of a system in terms of its organizing processes (organizers). To us, however, only some of these autonomous, organized systems, could be called \textit{self}-organized. We understand self-organization as the ability of a system to increase its entities' counts \textit{within a finite bound}, reaching a stable trajectory in its dynamics. This implies that in such a system, autosynthetic and autoanalytic processes coincide, without a complete `domination' of one or the other, as long as the system retains \textit{positive} counts. Arriving at this definition allows us to see how all of the individual properties described heretofore converge.

León-Portilla described how in Náhuatl cosmogony there were two primordial deities of creation, dual and complementary, named \textit{Ometecuhtli} and \textit{Omecihuatl}. He then proposed that they coincided into a single, dual deity, called \textit{Ōmeteōtl}, able to re-produce all there is, including itself \cite{portilla2006filosofia, wiercinski1984ometeotl}. Many other cultures coincide with similar dual aspects of `self-sustained' systems. We believe that the landscape of self-organization theories needs to adopt a similar duality, moving away from `autocatalysis'-only, into a more integral view of living systems, where the consumption/destruction of components is just as important as their production. We can show now how the notions we have developed so far crystallize into a \textit{self}-organized system.

\begin{definition}
    A \textbf{self-organized} system consists a system $\mathcal{S}$ together with a non-zero subsystem $0<\mathcal{A}\leq\mathcal{S}$, a full vector $n$, and vectors $u$ and $v$ with $u+v$ full, such that $\mathcal{A}n=0$, $\mathcal{A}^\top (u+v)=0$, $(\mathcal{S}^\top u)_a>0$, $(\mathcal{S}^\top v)_a<0$, and $\mathcal{S}^\top (u+v)=0$.
\end{definition}

As in the previous sections, we can show that there are processes that obey the vectorial definition of self-organization given. Let us then say that a process is \emph{self-organized} if it satisfies the following conditions: 
\begin{enumerate}
    \item Each input and output have an invalid match. 
    \item It is cyclic. 
    \item It contains a counts-increasing subprocess. 
    \item All inputs have a path to an output and all outputs have a path from an input. 
\end{enumerate}

We can prove this existence formally:

\begin{theorem}
    Each self-organized system has a compatible self-organized process.
\end{theorem}
\begin{proof}
    Let $(\mathcal{S}, \mathcal{A},n,u,v)$ be a self-organized system. Compose in parallel the transformations within $u+v$. This is sufficient for condition 3, since the parallel composition of $u$ is counts-increasing by definition. Proceed by making matches given by subsystem $\mathcal{A}$, until all matches left are invalid. This gives us condition 1. The process will be cyclic since we know $u+v$ is a cyclic vector by definition, so it satisfies condition 2. Finally, by Proposition~\ref{thm:84293018} all transformations have both an input and an output. This means that all paths begin and end in entity nodes, which gives us condition 4. 
\end{proof}

Thus, in simple words: a system is self-organized if it is able to autonomously synthesize itself, \textit{but also} controls this growth with a complementary self-consumption mechanism. We can add that a state of a system is self-organized if it contains the input of a self-organized process. 

\begin{figure}
    \centering
    \includegraphics[width=0.6\textwidth]{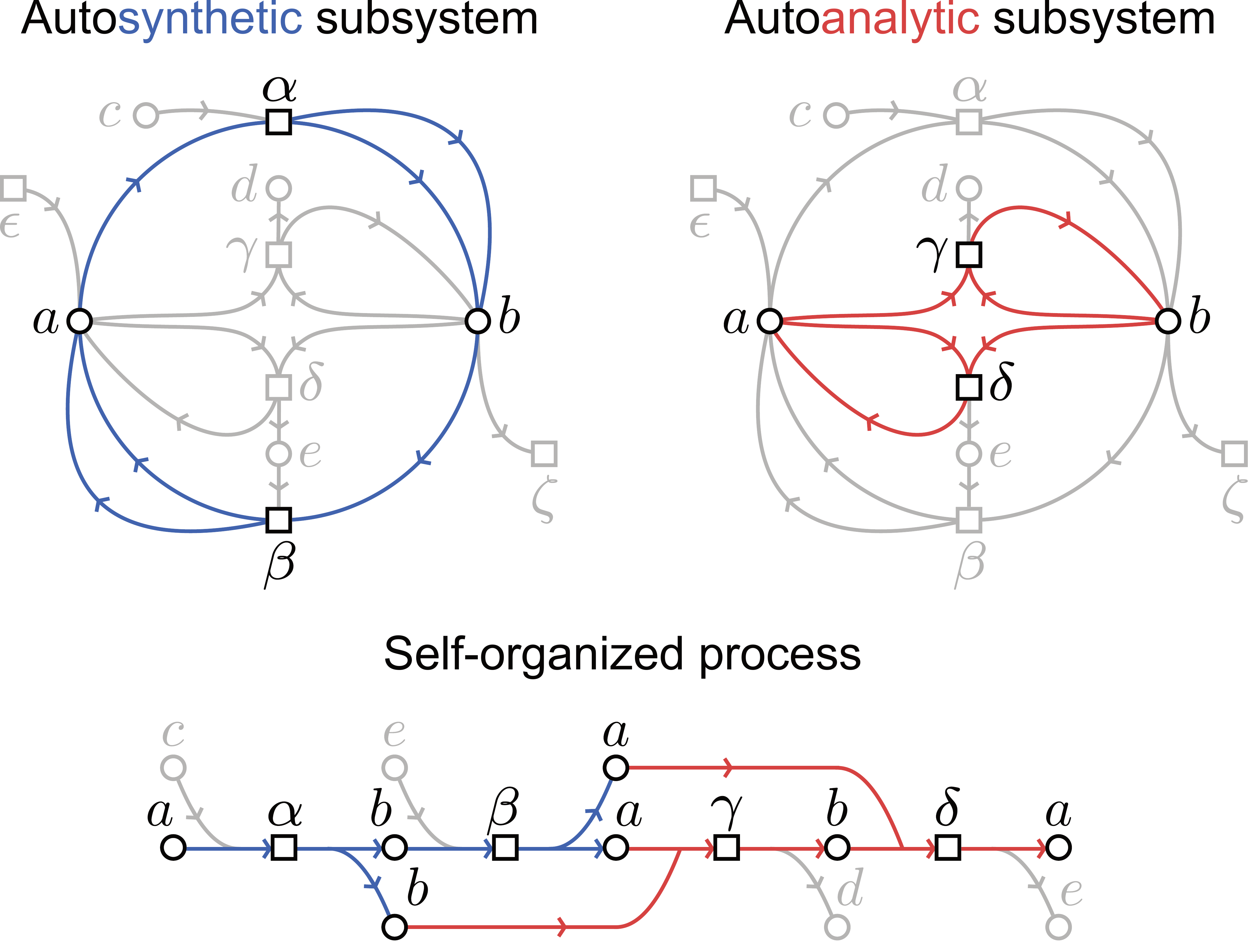}
    \caption{Autosynthetic (left) and autoanalytic (right) subsystems of the system graph. These systems can together be unfolded into a process that we can recognize as self-organized, as it meets the four criteria in the text.}
    \label{fig:autosynthesisanalysis}
\end{figure}

\subsection{System resilience}

Notice that when a self-organized system contains an autoanalytic component, this gives rise to the possibility that it may autonomously consume itself down to a state that loses the self-organization property. Let us refer to these states as \emph{sinks}. We can then say that a self-organized system is \emph{vulnerable} if there is a process starting from a self-organized state and ending in a sink. Otherwise, if a self-organized system is always able to return to a self-organized state, we say that it is \emph{resilient}. Thus, we need to characterize the nature of these sinks.

Let us consider states that contain \textit{enough} entities as to support the autosynthetic components of a system. If we find the minimal ensemble/state that satisfies this condition, then we can think of it as a \textit{seed} for the synthetic aspect of the system. This seed is then capable to regenerate all the components of the system. Let us define an \emph{autosynthetic state} as any state that contains the input of some autosynthetic processes. Then, we define a \emph{minimally-synthetic seed} (MSS) as the minimal of these autosynthetic states, for which no proper substate is itself autosynthetic. The proposition has a simple proof.

\begin{proposition}
    A state is a sink, if and only if, it does not contain a minimally-synthetic seed (MSS).
\end{proposition}
\begin{proof}
    Let $(\mathcal{S},\mathcal{A},n,u,v)$ be a self-organized system. Let $x$ be a state that is a sink. Then, all proceses starting at $x$ fail to be autoynthetic, as it would otherwise be possible to reach an arbitrarily large state and eventually become self-organized. Then $x$ cannot contain an MSS. If, conversely, $x$ does not contain an MSS, it may not reach a self-organized state. Otherwise, such self organized state would be autosynthetic, and some arbitrary growth would be possible to construct an autosynthetic state beginning at $x$, which would mean it contains an MSS. 
\end{proof}

We can now consider the dual states, these appearing as outputs of autoanalytic processes. Similarly, we consider the minimal states that satisfy this property. We define an \emph{autoanalytic state} as one that contains the output of an autoanalytic process. Then, a \emph{minimally-analytic seed} (MAS) is a state that is minimally autoanalytic, but for which no proper substate is itself autoanalytic. We can show the formal relationship between MSSs and MASs and its relevance to  resilience:

\begin{lemma}
    A self-organized system is resilient, if and only if every minimally-analytic seed (MAS) contains an MSS.
\end{lemma}
\begin{proof}
    Let $(\mathcal{S},\mathcal{A},n,u,v)$ be a self-organized system. First, assume that the system is resilient. Let $x$ be an MAS. Since $x$ is the output of an autoanalytic process, it is possible to reach $x$ from arbitrarily large states. Then, it is possible to reach $x$ from some self-organized state. Since the system is resilient, it is possible to reach a self-organized state $y$ from $x$. Then, it is possible to construct an autosynthetic process starting at $y$ to an arbitrarily large state. It is then possible to construct an autosynthetic process starting at $x$, in which case $x$ contains an MSS. Thus, every MAS contains an MSS. 
    
    Second, assume that the system has some MAS that is a sink. Since MASs are reachable via autoanalytic processes, the system is vulnerable. Then, conversely, if a system is vulnerable, a process exists from self-organized states to sink ones. Since MASs are minimal autoanalytic states, there is some MAS that is contained in some sink, and thus is itself a sink.
\end{proof}

\section{Concluding remarks}

Traditionally, researchers define an organism (and similarly, life itself) by compiling a list of features thought to be important for organisms and the ways they relate to their environment, and then looking at the kinds of natural systems that fulfill such desiderata. However, so far, such a strategy has not produced a satisfactory definition \cite{benner2010defining, tirard2010definition, tsokolov2009definition}. Here, we have taken a different approach, not to define life, but to formalize the kinds of systems that living organisms are, starting from very simple considerations. We described how the components of a system can undergo elemental transformations, which in turn unfold in time to give place to system-level features that we believe are essential to living systems: autonomy, self-organization, and resilience. We concluded by proposing that an autonomous, resilient self-organized system \textit{is a model of an organism}. This suggests a stronger hypothesis: \textit{all organisms realize a formal model of an autonomous, resilient, self-organized system}. In other words, we believe that however varied organisms are, they all belong to a natural `class' of systems, which \textit{admits} the kind of model here developed. Importantly, we are not stating that our proposed model is the unique model that organisms admit. 

First, we would like to mention a few topics for future work on the model itself, and then discuss broader implications. One topic is the explicit inclusion of the environment of an organism in its behavior. In our formalism, we deliberately excluded the environment as a center of discussion, in order to focus on the organism itself, but it is already implicit when we made a distinction between a system/subsystem and its complement (the environment, e.g., when discussing \textit{internal} entities/processes). Another interesting avenue of research is to explore biological dynamics that diverge from the same `initial condition', which is already evident in our model (\textit{e.g.}, when we say that matching of process graphs is not uniquely defined). This would open paths of inquiry for studying historical contingency, pre-adaptation dynamics, and evolvability.

The independence of our model from particular physical entities or specific organisms does not imply that physical constraints and properties are irrelevant. On the contrary, a natural step in this research program would be to find the instances in nature that \textit{embody} the model properties. The physical realization of a formal model --- the organisms themselves, here --- can certainly exhibit features not observed within the formal world alone. This is, essentially, what Rosen calls the `realization problem' of a model of living systems \cite{rosen1991life}, which is an ambitious but needed research goal that would give more ground to a coherent theory of organisms. We caution that such efforts will be, as Rosen points out, highly nontrivial, because in general synthesis is not the inverse of analysis: while molecular biology has been highly successful in isolating and cataloging all kinds of individual molecules inside organisms, nobody has ever built an organism starting from molecules. This suggests that the problem of life origins is of a different nature than the problem of defining what life is. Thus, developing a theory of organisms is only a partial step towards understanding life more generally. This work represents a small contribution towards this goal.

\section{Acknowledgments}
P.M.Z. acknowledges Judith Rosen, Stuart Kauffman, Joana Xavier, Carlos Gershenson, Ricard Solé, Otto Rössler, and Howard Pattee for insightful discussions about different aspects of this research. A.O.M. acknowledges Walter Fontana for discussions about the nature of autocatalysis and its dual, destructive autocatalysis. E.P.B. acknowledges support from the National Institutes of Health (U.S.A.) under grant numbers NIH T32GM142616 and R35GM138354. P.M.Z. and A.O.M. acknowledge support from the Santa Fe Institute and the Omidyar Complexity Fellowship.

\bibliographystyle{elsarticle-num} 
\bibliography{cas-refs}





\end{document}